\documentclass[11pt,a4paper,english]{amsart}

\usepackage{color}

\usepackage{babel}
\usepackage[latin1]{inputenc}
\usepackage{amsmath}
\usepackage{amsthm}
\usepackage{amsfonts}
\usepackage{indentfirst}
\usepackage{graphicx}
\usepackage{amssymb}
\usepackage[mathscr]{eucal}
\usepackage{tikz}
\usepackage{caption}
\usepackage{amsthm}
\usepackage{amscd}

\usepackage[T1]{fontenc}

\newtheorem{teo}{Theorem}
\newtheorem{pro}[teo]{Proposition}
\newtheorem{cor}[teo]{Corollary}
\theoremstyle{definition}
\newtheorem{rem}[teo]{Remark}

\newcommand{\ps}{{\perp_s}}

\title[Asymmetric EAQECCs and BCH codes]{Asymmetric entanglement-assisted quantum error-correcting codes and BCH codes}
\author{C. Galindo, F. Hernando, R. Matsumoto and D. Ruano}
\curraddr{\texttt{Carlos Galindo and Fernando Hernando:} Instituto
Universitario de Matem\'aticas y Aplicaciones de Castell\'on and
Departamento de Matem\'aticas, Universitat Jaume I, Campus de Riu
Sec. 12071 Castell\'{o} (Spain)\\
\texttt{Ryutaroh Matsumoto:} Department of Information and
Communication Engineering, Nagoya University, Nagoya, 464-8603 Japan,
and Department of Mathematical Sciences, Aalborg University,
  Denmark\\
\texttt{Diego Ruano:} IMUVA-Mathematics Research Institute, University of Valladolid, 47011 Valladolid (Spain).}
\email{
galindo@uji.es;
carrillf@uji.es;
ryutaroh.matsumoto@nagoya-u.jp;
diego.ruano@uva.es}
\date{}
\thanks{This work was supported in part by the Spanish Government MICINN/FEDER grants PGC2018-096446-B-C21, PGC2018-096446-B-C22 and RED2018-102583-T and MINECO grant RYC-2016-20208 (AEI/FSE/UE), Generalitat Valenciana, grant AICO-2019-223, as well as by Universitat Jaume I, grant P1-1B2018-10. Also by the JSPS (Japan) under grant 17K06419 and by the Junta de CyL (Spain) under grant VA166G18.}

\begin{document}

\begin{abstract}
The concept of asymmetric entanglement-assisted quantum error-correcting code (asymmetric EAQECC) is introduced in this article. Codes of this type take advantage of the asymmetry in quantum errors since phase-shift errors are more probable than qudit-flip errors. Moreover, they use  pre-shared entanglement between encoder and decoder to simplify the theory of quantum error correction and increase the communication capacity. Thus,  asymmetric EAQECCs can be constructed from any pair of classical linear codes over an arbitrary field. Their parameters are described and a Gilbert-Varshamov bound is presented. Explicit parameters of  asymmetric EAQECCs from BCH codes are computed and examples exceeding the introduced Gilbert-Varshamov bound are shown. \end{abstract}

\maketitle

\section{Introduction}
In the last decades the interest in quantum computation has grown exponentially, mainly because it transforms some intractable problems into tractable ones as showed the polynomial time algorithms given by Shor for discrete logarithms and prime factorization \cite{22RBC}.

The usage of subatomic particles to hold memory and the application of quantum mechanics determine the behavior of quantum computers. These computers (the current implementations) are less reliable than the classical ones and produce more errors. Another inconvenient with this computers is decoherence and, even when one cannot clone quantum information \cite{8AS, 26RBC}, both challenges can be addressed with quantum error correction \cite{23RBC, 95kkk}.

The first steps in the construction of quantum error-correcting codes corresponded to the binary case \cite{18kkk, 19kkk, 38kkk} (see also \cite{7kkk, 8kkk, 45kkk}). Afterwards and especially because of their interest in fault-tolerant computation the non-binary case was also studied \cite{ketkar06} (some more references are \cite{AK,BE,opt,lag3, 71kkk}). Most of the quantum error-correcting codes are stabilizer codes where the error group is determined by eigenspaces with eigenvalue $1$.

Sufficient (respectively, necessary) conditions for existence of (sometimes pure) quantum codes are given by the Gilbert-Varshamov bounds \cite{eck,feng, ketkar06, matsumotouematsu01} (respectively, quantum singleton or Hamming bounds \cite{3DC, opt,ketkar06, 28DC}).

Unitary operators, usually denoted $X$ and $Z$, are used to provide  quantum (error-correcting) codes and the minimum distance $d$ of such codes indicates that one can correct up to $\lfloor (d-1)/2 \rfloor$ phase-shift and qudit-flip errors. In \cite{IOF}, the authors noticed that phase-shift errors happened more likely than qudit-flip errors, thus it was desirable to construct quantum codes where two minimum distances $d_x$ and $d_z$, for detecting qudit-flip and phase-shift errors, respectively, were considered and provide results for addressing their behavior. As a consequence, in the last years asymmetric quantum error-correcting codes have been studied giving rise to codes suitable when dephasing occurs more often than relaxation \cite{EZ1,EZ2,EZ3,LG2,LG3,Sarve}. Most of the asymmetric quantum codes come from the CSS construction of quantum stabilizer codes and, for them, there is also a Gilbert-Varshamov bound \cite{Matsu}. In addition, the existence of an asymmetric quantum error-correcting code coming from the CSS construction can also be applied to linear ramp secret sharing and communication over wiretap channels of type II \cite{GeG}.

To provide an asymmetric (or symmetric) quantum code requires some type of self-orthogonality of the classical constituent code (or an inclusion of a constituent code into the dual of other constituent one) and, then, many good classical codes cannot be considered for that purpose. For overcoming this restriction and boosting the rate of transmission, it was proposed in \cite{brun1} (for the symmetric case) to share entanglement between encoder and decoder. Some constructions of this type for binary codes (and also for codes over finite fields $\mathbb{F}_p$, $p$ prime) can be found in the literature \cite{hsieh,schin,wilde08}. The case when the codes are supported in an arbitrary finite field has been described in \cite{GaHMR}.

It seems clear that it remains to consider entanglement-assisted quantum error-correcting codes (EAQECCs) for the asymmetric case. To the best of our knowledge this task had not been performed yet. Section \ref{sect2} of this paper is devoted to explain how to construct and which are the parameters of an asymmetric EAQECC obtained from any two linear classical codes. Theorem \ref{el4} and Theorem \ref{puncturing} (for nested constituent codes) are the main results in this section.  Section \ref{sect3} gives a Gilbert-Varshamov bound for asymmetric EAQECCs; we state and prove this bound for both the finite and the asymptotic case. In Section \ref{sect4} we present the explicit computation of the parameters  of asymmetric EAQECCs coming from BCH codes, see Theorem \ref{elB} and Corollary \ref{elc}. 
 Finally, our Section  \ref{sect5} provides examples of asymmetric EAQECCs which exceed the Gilbert-Varshamov bound before stated. Notice that asymmetric EAQECCs give rise to (symmetric) EAQECCs and in this section we show also examples of EAQECCs obtained with our procedure exceeding the Gilbert-Varshamov bound for EAQECCs.

\section{Asymmetric EAQECCs}
\label{sect2}
Let $q=p^r$ a positive power of a prime number $p$ and set $\mathbb{F}_q$ the finite field of order $q$. A $q$-ary {\it stabilizer quantum code} is the linear space of $(\mathbb{C}^q)^n$ given by the intersection of the eigenspaces with eigenvalue 1 corresponding to some subgroup $S$ of the error group $G_n$ generated by the matrices corresponding to a basis of $ \mathrm{Hom} \left(
(\mathbb{C}^q)^{\otimes n}, (\mathbb{C}^q)^{\otimes n}\right)$, that is $G_n$ is determined by the product $X(\mathbf{a}) Z(\mathbf{b})$ of  tensor products $X(\mathbf{a})= X(a_1) \otimes X(a_2) \otimes \cdots \otimes X(a_n)$ and $Z(\mathbf{b})= Z(b_1) \otimes Z(b_2) \otimes \cdots \otimes Z(b_n)$ of unitary operators $X$ and $Z$ over $\mathbb{C}^q$, where $\mathbf{a} = (a_1, a_2, \ldots, a_n)$, $\mathbf{b} = (b_1, b_2, \ldots, b_n)$, and $a_i, b_i \in \mathbb{F}_q$, $1 \leq i \leq n$. It is known \cite[Lemma 11]{ketkar06} that an error in $G_n$ is detectable by the stabilizer code if and only if it belongs to the group generated by the subgroup $S$ and the center of $G_n$ or the error is not in the centralizer of $S$ in $G_n$.

The above facts can be regarded in terms of additive codes in $\mathbb{F}_q^{2n}$. In order to do this, we introduce the trace-symplectic form for two vectors $\left(\mathbf{a}|\mathbf{b}\right), \left(\mathbf{a'}|\mathbf{b'}\right) \in \mathbb{F}_q^{2n}$ as follows:
\[
\left(\mathbf{a}|\mathbf{b}\right) \cdot_{ts} \left(\mathbf{a'}|\mathbf{b'}\right) = \mathrm{tr}_{q|p} \left(\mathbf{a} \cdot \mathbf{b'} - \mathbf{a'} \cdot \mathbf{b}\right) \in \mathbb{F}_p,
\]
where  $\mathrm{tr}_{q|p}$ is the trace map and $\cdot$ the inner product in $\mathbb{F}_q^{n}$. Then (in the linear case) an $[[n,k, d]]_q$ stabilizer quantum code exists if and only if there is a linear code $C \subseteq \mathbb{F}_q^{2n}$ of dimension $n-k$ such that $C \subseteq C^{\perp_{ts}}$, where $C^{\perp_{ts}}$ stands for the dual code with respect to the $\cdot_{ts}$ product. Here the minimum distance $d$ is determined by the minimum symplectic weight $\mathrm{swt} ( C^{\perp_{ts}} \setminus C)$. It is convenient to recall that for $\left(\mathbf{a}|\mathbf{b}\right)$ as above,
\[
\mathrm{swt} \left(\mathbf{a}|\mathbf{b}\right) = \# \left\{j \in \{1, 2, \ldots, n\} | (a_j,b_j) \neq (0,0)\right\},
\]
$\#$ meaning cardinality.

A particular case in the above construction follows from the so-called CSS (Calderbank-Shor-Steane) procedure \cite{95kkk, 19kkk}. Here we need two linear codes $C_1$ and $C_2$ in $\mathbb{F}_q^{n}$ such that $C_2 \subseteq C_1^\perp$, $\perp$ means Euclidean duality, and then the code $C=C_1 \times C_2 \subseteq \mathbb{F}_q^{2n}$ provides a stabilizer quantum code whose parameters depend on those of $C_1$ and $C_2$. Some classical references are \cite{AK,BE, 18kkk, 19kkk, 20kkk}.

The fact that dephasing usually happens much more often that relaxation  \cite{IOF} motivated the study and searching of asymmetric quantum error-correcting codes \cite{EZ1,EZ2,EZ3, LG1,LG2,LG3,LG4}. For this purpose, the most used procedure is the CSS construction because it easily allows us to get parameters $d_z$ and $d_x$ such that our previous stabilizer code detects phase-shift (respectively, qudit-flip) errors up to weight $d_z-1$ (respectively, $d_x -1$). The specific result (see \cite[Lemma 3.1]{Sarve}) states that

\begin{teo}
Let $C_2 \subset C_1^{\perp} \subseteq \mathbb{F}_q^{n}$ be linear codes. The CSS construction gives rise to an asymmetric quantum code with parameters $[[n, \dim C_1^\perp - \dim C_2, d_z/d_x]]_q$, where $d_z$ (respectively, $d_x$) is the minimum Hamming weight of the set $C_1^\perp \setminus (C_2 \cap C_1^\perp)$ (respectively, $C_2^\perp \setminus (C_1 \cap C_2^\perp)$).
\end{teo}

The previously mentioned stabilizer and asymmetric quantum codes require self-orthog-onality conditions with respect to trace-symplectic duality and not every classical linear code can be used for providing those quantum codes. The self-orthogonality condition can be bypassed if encoder and decoder share some quantity of entanglement \cite{brun1} giving rise to the so called entanglement-assisted quantum error-correcting codes (EAQECCs). In the binary case the construction of these codes is described in \cite{hsieh} (third paragraph of Section II). This construction also holds for codes over finite fields of the type $\mathbb{F}_p$, $p$ being a prime number (see \cite[Remark 1]{wilde08} and \cite{schin} for a proof). There it is proved that one can obtain an EAQECC from a classical code $C \subseteq \mathbb{F}_p^{2n}$  such that $C \not\subseteq C^{\perp_{ts}}$ and the set of detectable quantum errors is given by
\[
\left(C \cap C^{\perp_{ts}}  \right) \cup \left(\mathbb{F}_p^{2n} \setminus  C^{\perp_{ts}} \right).
\]

On $\mathbb{F}_q^{2n}$, $q=p^r$, one can also define a symplectic product:
\[
\left(\mathbf{a}|\mathbf{b}\right) \cdot_{s} \left(\mathbf{a'}|\mathbf{b'}\right) = \left(\mathbf{a} \cdot \mathbf{b'} - \mathbf{a'} \cdot \mathbf{b}\right) \in \mathbb{F}_q.
\]
Using a suitable basis of $\mathbb{F}_q$ over $\mathbb{F}_p$, an isomorphism of $\mathbb{F}_p$-linear spaces $\phi: \mathbb{F}_p^{2r} \rightarrow \mathbb{F}_q^{2}$ can be given, providing an isomorphism of $\mathbb{F}_p$-linear spaces
\[
\phi^E: \left(\mathbb{F}_p^{r}\right)^n  \times \left(\mathbb{F}_p^{r}\right)^n \rightarrow \mathbb{F}_q^{2n}.
\]
With the help of $\phi^E$, in \cite{GaHMR},  the results of EAQECCs over  $\mathbb{F}_p$ can be extended to $\mathbb{F}_q$ and the product $\cdot_s$ instead of $\cdot_{ts}$. Indeed, the following result holds:

\begin{teo}\label{el2}
Let $C \subseteq \mathbb{F}_q^{2n}$ be a linear code over $\mathbb{F}_q$ of dimension $(n-k)$. Denote by $H = (H_X|H_Z)$ a generator matrix for $C$. Let $C' \subseteq \mathbb{F}_q^{2(n+c)}$ be a linear code over $\mathbb{F}_q$ whose projection to the coordinates $1, 2, \ldots, n, n+c+1, n+c+2, \ldots, 2n+c$ equals $C$ and such that $C' \subseteq (C')^\ps$,  $c$ being the minimum required number of maximally entangled quantum states in $\mathbb{C}^q \otimes \mathbb{C}^q$. Then,
$$
2c = \mathrm{rank}\left(H_X H_Z^T - H_Z H_X^T\right) = \dim_{\mathbb{F}_q} C - \dim_{\mathbb{F}_q} \left(C \cap C^\ps\right).
$$

The encoding quantum circuit is constructed from $C'$, and it encodes $k+c$ logical qudits in $\mathbb{C}^q \otimes \cdots (k+c\; \mbox{times}) \cdots \otimes \mathbb{C}^q$ into $n$ physical qudits using $c$ maximally entangled pairs. The minimum distance is
$$
d= d_s\left(C^\ps \setminus (C\cap C^\ps)\right) = \min\left\{ \mathrm{swt}\left(\mathbf{a}|\mathbf{b}\right) \mid \left(\mathbf{a}|\mathbf{b}\right) \in C^\ps \setminus (C\cap C^\ps)\right\}.
$$

As a consequence,   $C$ provides an $[[n,k+c,d;c]]_q$ EAQECC over the field $\mathbb{F}_q$.
\end{teo}

In this paper we are interested in the asymmetric case and we desire to construct asymmetric EAQECCs from two linear codes $C_1$ and $C_2$ over an arbitrary finite field $\mathbb{F}_q$. Assume that $H_1$ (respectively, $H_2$) is a generator matrix of $C_1$ (respectively, $C_2$).

The above described construction of stabilizer codes over $\mathbb{F}_p$ following the CSS procedure determines asymmetric EAQECCs coming from any two linear codes $C_1, C_2 \subseteq \mathbb{F}_p^n$. Here the code $C$ over $\mathbb{F}_p^{2n}$ is $C=C_1 \times C_2$ and $C^{\perp_s} = C_2^\perp \times C_1^\perp$, where $\perp$ denotes the Euclidean dual. Notice that in this case $\cdot_{ts} = \cdot_s$.
The set of detectable errors is
\begin{multline*}
\left((C_1 \cap C_2^\perp) \times (C_2 \cap C_1^\perp)\right) \cup
  \left(\mathbf{F}_p^{2n} \setminus
  C_2^\perp \times C_1^\perp\right)\\
   = \left((C_1 \cap C_2^\perp) \cup (\mathbf{F}_p^n \setminus C_2^\perp)\right)   \times
  \left((C_2 \cap C_1^\perp) \cup (\mathbf{F}_p^n \setminus C_1^\perp)\right).
\end{multline*}
Defining
\begin{equation}
\label{distances}
  d_z = \mathrm{wt}\big(C_1^\perp \setminus (C_2 \cap C_1^\perp)\big) 
  \mbox{and }  d_x = \mathrm{wt}\big(C_2^\perp \setminus (C_1 \cap C_2^\perp)\big),
\end{equation}
where $\mathrm{wt}$ means minimum Hamming weight, it is clear we are able to construct an asymmetric EAQECC which can detect up to $d_x - 1$ qudit-flip errors and
up to $d_z - 1$ phase-shift errors.

These results can be extended to any finite field $\mathbb{F}_q$ using again the above described isomorphism $\phi^E$ and \cite[Proposition 1]{GaHMR} which relates $\cdot_{st}$ and $\cdot_s$. The general result being:
\begin{teo}
\label{el4}
Consider linear codes $C_i \subseteq \mathbb{F}_q^{n}$ of dimension $k_i$ and generator matrix $H_i$, $i=1,2$. Set $d_x$ and $d_z$ as in (\ref{distances}).

Then $C_1 \times C_2 \subseteq \mathbb{F}_q^{2n}$ gives rise to an asymmetric EAQECC which encodes $n-k_1-k_2 + c$ logical qudits into $n$ physical qudits which can correct  up to $\lfloor (d_x - 1)/2 \rfloor$ qudit-flip  errors and up to $\lfloor (d_z - 1)/2 \rfloor$ phase-shift errors. The minimum required of maximally entangled pairs is
\[
c = \mathrm{rank}(H_1H_2^T) =  \dim C_1 - \dim (C_1 \cap C_2^\perp).
\]

As a consequence, we obtain an $$[[n, n-k_1-k_2+c,d_z/d_x;c]]_q$$ asymmetric EAQECC.
\end{teo}

We end this section with a result that assumes that our constituent linear codes are nested. We will see that the asymmetric EAQECC comes from puncturing a code in $\mathbb{F}_q^{2n}$.

\begin{teo}
\label{puncturing}
Let $C_1$ and $C_2$ be $\mathbb{F}_{q}$-linear codes such that
$C_2 \subseteq C_1 \subseteq \mathbb{F}_{q}^n$. Set $k_i = \dim C_i$, $i \in \{1,2\}$ and $d_1^\perp$ (respectively, $d_2$) the minimum distance of the code $C_1^\perp$ (respectively, $C_2$). Suppose that $c$ is a positive integer such that it satisfies $1 \leq c \leq \min \{d_1^\perp, d_2\} -1$. Then, there exists an asymmetric EAQECC with parameters
$$
[[n-c, k_1-k_2+c, d_z/d_x; c]]_q,
$$
where $d_z$ (respectively, $d_x$) is the minimum Hamming weight of the elements in the set $C_2^\perp \setminus C_1^\perp$ (respectively, $C_1 \setminus C_2$).
\end{teo}
\begin{proof}
Consider the code $C= C_2 \times C_1^\perp$, then $C \subseteq C^{\perp_s} = C_1 \times C_2^\perp$ and $2c \leq \mathrm{wt} (C \setminus \mathbf{0})-1$, and the result follows from \cite[Theorem 7]{GaHMR} and (\ref{distances}).

Notice that the above asymmetric EAQECC comes from the punctured code defined as
\[
P(C) = \big\{\left(\mathrm{pr}(\mathbf{a}), \mathrm{pr}(\mathbf{b}\right)\; | \; (\mathbf{a},\mathbf{b}) \in C \big\},
\]
$\mathrm{pr}$ being the projection to the first $n-c$ coordinates. In fact, according to the proof of \cite[Theorem 9]{GaHMR}
\[
\dim P(C) - \dim \left(P(C) \cap P(C)^{\perp_s}\right) =2c,
\]
which by Theorem \ref{el2} shows that $c$ is the number of maximally entangled pairs.
\end{proof}

\section{A Gilbert-Varshamov bound for asymmetric EAQECCs}
\label{sect3}
We devote this section to provide a finite and an asymptotic Gilbert-Varshamov-type (GV) bound for asymmetric EAQECCs. We start with the finite case.
\subsection{The finite GV bound}
Let us start with our result.
\begin{teo}
\label{FGV}
Consider positive integer numbers $n, k_1, k_2, d_z,$ $d_x$ and $c$ such that $k_1 \leq n$, $k_2 \leq n$ and
\[
k_1 +k_2-n \leq c \leq \min \{k_1, k_2\},
\]
which satisfy the following inequality
$$
\frac{q^{n-k_1} - q^{k_2 - c}}{q^n-1} \sum_{i=1}^{d_z-1} {n \choose i}(q-1)^i   +
\frac{q^{n-k_2} - q^{k_1 - c}}{q^n-1}\sum_{i=1}^{d_x-1} {n \choose i}(q-1)^i < 1,
$$
then there exists an $[[n,n-k_1-k_2+c,d_z/d_x;c]]_q$ asymmetric EAQECC.
\end{teo}
\begin{proof}
For simplicity sake, in this proof $C'_2$ will be used instead of $C_2^\perp$.
Consider integer numbers $n,k_1,k_2$ and $c$ as in the statement. Define
\begin{multline*}
A(n,k_1,k_2,c) =\big\{ (C_1, C'_2) \; | \;
C_1, C'_2 \subset \mathbf{F}_q^n, \\
\dim C_1 = k_1, \dim C'_2 = n-k_2, \mbox {and } 
c = \dim C_1 - \dim (C_1 \cap C'_2) \big\}.
\end{multline*}
For $\mathbf{v} \in \mathbf{F}_q^n$, define also
$$
B_{z}(\mathbf{v}) = \big\{ (C_1, C'_2) \in A(n,k_1,k_2,c) \; | \; 
\mathbf{v} \in C_1^\perp \setminus (C^{\prime\perp}_2 \cap C_1^\perp) \big\}
$$
and
$$
B_{x}(\mathbf{v}) = \big\{ (C_1, C'_2) \in A(n,k_1,k_2,c) \; | \; 
\mathbf{v} \in C'_2 \setminus (C_1 \cap C'_2) \big\}.
$$
For nonzero $\mathbf{v}_1$ and $\mathbf{v}_2 \in \mathbf{F}_q^n$,
we claim that
$$\# B_{z}(\mathbf{v}_1)=\# B_{z}(\mathbf{v}_2),$$ where we recall that $\#$ means cardinality.

Let us see a proof. Denote by $GL(n,q)$ the set of invertible matrices on $\mathbf{F}_q^n$ and for a fixed $(D_1, D'_2) \in A(n,k_1,k_2,c)$,
a fixed $M_1 \in GL(n,q)$ with $M_1 \mathbf{v}_1 = \mathbf{v}_2$
and $M'_1  \in GL(n,q)$ with $M'_1 \mathrm{span}(\mathbf{v}_1)^\perp
= \mathrm{span}(\mathbf{v}_2)^\perp$, where $M'_1 \mathrm{span}(\mathbf{v}_1)^\perp$ stands for the linear space given by the products $M'_1 \mathbf{w}$ such that $\mathbf{w} \in \mathrm{span}(\mathbf{v}_1)^\perp$. Then we have
\begin{eqnarray*}
  && \# B_{z}(\mathbf{v}_1)\\
  &=& \# \big\{ (C_1, C'_2) \in A(n,k_1,k_2,c) |  \mathbf{v}_1 \in C_1^\perp \setminus (C^{\prime\perp}_2 \cap C_1^\perp)\big\}\\
  &=& \# \big\{ (C_1, C'_2) \in A(n,k_1,k_2,c) |
\mathrm{span}(\mathbf{v}_1)^\perp \supseteq C_1   \mbox{ and } \mathrm{span}(\mathbf{v}_1)^\perp
\not\supseteq C'_2 \big\}\\
  &=& \# \big\{  (M D_1, MD'_2) \;| \;
\mathrm{span}(\mathbf{v}_1)^\perp \supseteq MD_1   \mbox{ and } \mathrm{span}(\mathbf{v}_1)^\perp
\not\supseteq MD'_2, M \in GL(n,q) \big\}\\
  &=& \# \big\{  (M'_1M D_1, M'_1MD'_2)|
M'_1 \mathrm{span}(\mathbf{v}_1)^\perp  \supseteq \! M'_1MD_1 \\
&& \mbox{ and } M'_1 \mathrm{span}(\mathbf{v}_1)^\perp
\not\supseteq M'_1MD'_2, M \in GL(n,q) \big\}\\
  &=& \# \big\{  (M'_1M D_1, M'_1MD'_2) |
\mathrm{span}(\mathbf{v}_2)^\perp \supseteq M'_1M D_1 \\
&& \mbox{ and } \mathrm{span}(\mathbf{v}_2)^\perp \not\supseteq M'_1MD'_2, M'_1M \in GL(n,q) \big\}\\
&=& \# B_{z}(\mathbf{v}_2).
\end{eqnarray*}
We also claim that
$\# B_{x}(\mathbf{v}_1)=\# B_{x}(\mathbf{v}_2)$. Indeed,
\begin{eqnarray*}
  && \# B_{x}(\mathbf{v}_1)\\
  &=& \# \big\{ (C_1, C'_2) \in A(n,k_1,k_2,c) |
\mathbf{v}_1 \in C'_2 \setminus (C^{\prime}_2 \cap C_1) \big\}\\
  &=& \# \big\{ (MD_1, MD'_2)  \;| \;
\mathbf{v}_1 \in MD'_2 \setminus (MD^{\prime}_2 \cap MD_1),   M \in GL(n,q) \big\}\\
  &=& \# \big\{ (M_1MD_1, M_1MD'_2) |  
M_1\mathbf{v}_1 \in M_1MD'_2 \setminus (M_1MD^{\prime}_2 \cap M_1MD_1),\\&& M_1M \in GL(n,q) \big\}\\
  &=& \# \big\{ (MD_1, MD'_2)  \;| \;
M_1\mathbf{v}_1 \in MD'_2 \setminus (MD^{\prime}_2 \cap MD_1),   M \in GL(n,q) \big\}\\
&=& \# B_{x}(\mathbf{v}_2).
\end{eqnarray*}

Next we will count the quantity of triples $(\mathbf{v}$, $C_1$, $C'_2)$
such that $\mathbf{v} \in C_1^\perp \setminus (C^{\prime\perp}_2 \cap C_1^\perp)$
in two different ways. From \cite[Proposition 4]{GaHMR} and the fact that the rank of a matrix and its transpose coincide, we deduce that
\[
c = \dim C_1 - \dim \left( C_1 \cap C_2^\perp \right) = \dim C_2 - \dim \left( C_2 \cap C_1^\perp \right).
\]
Then, we observe that
$$
\dim C^{\prime\perp}_2 \cap C_1^\perp
= \dim C_2 \cap C_1^\perp \\=
\dim C_2 - (\dim C_2 - \dim C_2 \cap C_1^\perp) = k_2 - c.
$$
For each pair $(C_1, C'_2) \in A(n,k_1,k_2,c)$ there are
\[
q^{n-k_1} - q^{k_2 - c}
\]
vectors $\mathbf{v}$ such that $\mathbf{v} \in C_1^\perp \setminus (C^{\prime\perp}_2 \cap C_1^\perp)$.
Thus the total number of such triples is
\[
(q^{n-k_1} - q^{k_2 - c}) \# A(n,k_1,k_2,c).
\]

On the other hand, we can count the total number of triples as
\[
\sum_{\mathbf{0}\neq \mathbf{w} \in \mathbf{F}_q^n}
\# B_{z}(\mathbf{w}) = (q^n-1) \# B_{z}(\mathbf{v})
\]
for any fixed nonzero $\mathbf{v}$.
This implies
\[
\frac{\# B_{z}(\mathbf{v})}{\# A(n,k_1,k_2,c)}
= \frac{q^{n-k_1} - q^{k_2 - c}}{q^n-1}.
\]

A similar argument shows
\[
\frac{\# B_{x}(\mathbf{v})}{\# A(n,k_1,k_2,c)}
= \frac{q^{n-k_2} - q^{k_1 - c}}{q^n-1}.
\]

If we remove a pair $(C_1, C'_2)$ from
$A(n,k_1,k_2,c)$ either when
$\mathbf{v}_z \in C_1^\perp \setminus (C^{\prime\perp}_2 \cap C_1^\perp)$ or
when $\mathbf{v}_x \in C^{\prime}_2 \setminus (C^{\prime}_2 \cap C_1)$
for $1 \leq \mathrm{wt}(\mathbf{v}_z) \leq d_z-1$ and
for $1 \leq \mathrm{wt}(\mathbf{v}_x) \leq d_x-1$, then  we remove in total
\begin{equation}
\sum_{1 \leq \mathrm{wt}(\mathbf{v}_z) \leq d_z-1} \# B_z(\mathbf{v}_z) +
\sum_{1 \leq \mathrm{wt}(\mathbf{v}_x) \leq d_x-1} \# B_x(\mathbf{v}_x) \label{eq1}
\end{equation}
pairs from $A(n,k_1,k_2,c)$.

As a consequence, there exists at least one
$$[[n,n-k_1-k_2+c,d_z/d_x;c]]_q$$ asymmetric EAQECC whenever the number (\ref{eq1}) is less than $\# A(n,k_1,k_2,c)$ which proves the statement.
\end{proof}

\subsection{The asymptotic GV bound}
From Theorem \ref{FGV} and \cite{matsumotouematsu01}, it can be deduced the following asymptotic GV bound.
\begin{teo}
Consider positive real numbers $K_1, K_2, \delta_z,\delta_x$ and $\lambda$ such that
\[
K_1+K_2 -1 \leq \lambda \leq \min \{K_1,K_2\}.
\]
Set
$h_q(y) :=  -y \log_q y -(1- y) \log_q (1-y)$ the $q$-ary entropy function. If the inequalities
\begin{eqnarray*}
 h_q(\delta_z) + \delta_z \log_q(q-1) &<& K_1 \; \; \mbox{ and } \\
    h_q(\delta_x) + \delta_x \log_q(q-1) &<& K_2
\end{eqnarray*}
hold, then, for sufficiently large $n$, there exists an asymmetric EAQECC with parameters
$$\big[\big[n, \lfloor n-nK_1-nK_2+n\lambda\rfloor, \lfloor n\delta_z\rfloor/\lfloor n\delta_x\rfloor; \lfloor n\lambda \rfloor \big]\big]_q.$$
\end{teo}

\section{Asymmetric EAQECC from BCH codes}
\label{sect4}

The aim of this section is the construction of asymmetric EAQECCs with good parameters by using the results in Section \ref{sect2}. To carry it out, we consider specific BCH codes. Instead of the classical way, our BCH codes are regarded as subfield-subcodes of evaluation codes defined by evaluating univariate polynomials \cite{cas}. We consider this construction because it can be extended to evaluation by polynomials in several variables \cite{GaH,GaGHR} which we hope will give better codes in the future.

Let $\ell$ be a positive integer such that $r$ divides $\ell$ and consider a positive integer $N$ such that $N-1$ divides $p^\ell -1$. In this section, we use classes of univariate polynomials in the quotient ring $\mathbb{F}_{p^\ell} [X]/I$, where $I$ is the ideal of $\mathbb{F}_{p^\ell}[X]$ generated by $X^{N-1} -1$. If $Z=\{P_1, P_2, \ldots, P_n\}$, where $n=N-1$, is the zero set of $I$ in $\mathbb{F}_{p^\ell}$, we define the  evaluation map $$\mathrm{ev}: \mathbb{F}_{p^\ell}[X]/I \rightarrow \mathbb{F}_{p^\ell}^{n} \;\; \mathrm{ev}(f) = \left(f(P_1), f(P_2), \ldots, f(P_{n}) \right).$$

Assume that $\Delta$ is a subset of $\mathcal{H}:=\{0, 1, \ldots, N-2\}$, then we write $E_\Delta$ the code in $\mathbb{F}_{p^\ell}^n$ generated by the vectors $$\left\{ \mathrm{ev}\left(X^i\right) \; | \; i \in \Delta \right\}.$$

Within the congruence ring $\mathbb{Z}_{N-1}$, we consider minimal cyclotomic cosets with respect to $q=p^r$; minimal means that it contains exactly the elements of the form $aq^t$, $t \geq 0$ in  $\mathbb{Z}_{N-1}$ for some fixed element $ a \in \mathbb{Z}_{N-1}$ under the identification $\mathbb{Z}_{N-1} = \mathcal{H}$. Pick a representative $a$ (the least one) of each minimal cyclotomic coset which we denote $\mathfrak{I}_a$. Then $\{ \mathfrak{I}_{a}\}_{ \in \mathcal{A}}$ is the set of minimal cyclotomic cosets with respect to $q$, $\mathcal{A}$ being the set of representatives above mentioned. In addition set $i_{a} : = \# (\mathfrak{I}_{a})$. For convenience, we will write $$\mathcal{A} = \{a_0 = 0  < a_1 < a_2 < \cdots \} = \{a_j\}_{j=0}^z.$$

We will use the following two results which can be found in \cite{GaH, GaGHR}.

\begin{pro}
\label{la7}
Set $\Delta = \cup_{j=t'}^t  \mathfrak{I}_{a_j}$, $t' < t$. Then the subfield-subcode of $E_\Delta$ over $\mathbb{F}_q$,
$$
E_\Delta |_{\mathbb{F}_{q}} := E_\Delta \cap (\mathbb{F}_{q})^n,
$$
has dimension $\sum_{j=t'}^{t}  i_{a_j}$.
\end{pro}

\begin{pro}
\label{la8}
The minimum distance of the (Euclidean) dual of the subfield-subcode $E_\Delta |_{\mathbb{F}_{q}}$, where $\Delta = \cup_{j=0}^t  \mathfrak{I}_{a_j}$  is larger than or equal to $a_{t+1}+1$ (BCH bound).
\end{pro}

Next we state the main result in this section.
\begin{teo}
\label{elB}
With the above notation consider two different indices $s, t \in \{0,1, \ldots, z\}$ and assume that $s <t$. Then we can construct an asymmetric EAQECC with parameters
\[
\bigg[\bigg[n, n - \sum_{j=0}^{t} i_{a_j}, (a_{t+1}+1/a_{s+1}+1); \sum_{j=0}^{s} i_{a_j}\bigg]\bigg]_{q}.
\]
\end{teo}
\begin{proof}
Consider the linear codes $C_i= E_{\Delta_i} |_{\mathbb{F}_{q}}$, $i=1,2$, where $\Delta_1 = \bigcup_{j=0}^t  \mathfrak{I}_{a_j}$ and $\Delta_2 = \bigcup_{j=0}^s  \mathfrak{I}_{a'_j}$, $a'_j$ being the representative of the minimal cyclotomic coset containing $N-1-a_j$. Taking into account that $\# \mathfrak{I}_{a'_j} = \# \mathfrak{I}_{a_j}$, by Proposition \ref{la7} it holds that $k_1 := \dim C_1 = \sum_{j=0}^{t} i_{a_j}$ and  $k_2 := \dim C_2 = \sum_{j=0}^{s} i_{a_j}$.

Now $C_2^\perp = E_{\Delta'}|_{\mathbb{F}_{q}}$, where $\Delta' = \mathcal{H} \setminus \bigcup_{j=0}^s \mathfrak{I}_{a_j}$ \cite{GaGHR}. Hence, the minimum required of maximally entangled pairs is
$$
  c= \dim C_1 - \dim (C_1 \cap C_2^\perp)  \\
  =  \sum_{j=0}^{t} i_{a_j} -  \sum_{j=s+1}^{t} i_{a_j} = \sum_{j=0}^{s} i_{a_j}.
$$
The minimum distance of the dual codes satisfies $d (C_1^\perp) \geq a_{t+1} +1$ (by Proposition \ref{la8}) and $d (C_2^\perp) \geq a_{s+1}+1$ because $C_2$ contains $s+1$ consecutive cyclotomic cosets and it is equivalent to a code as in Proposition \ref{la8}.

Finally, applying Theorem \ref{el4}, we get an  asymmetric EAQECC with parameters as in the statement.
\end{proof}

From the previous result, we can deduce the following one.

\begin{cor}
\label{elc}
Keeping the above notation where $q=p^r$, assume that
\begin{equation}
\label{c1}
\left( p^r \right)^{\lfloor \frac{\ell}{2r}\rfloor} < n \leq p^\ell -1
\end{equation}
and pick and index $t$ such that
\begin{equation}
\label{c2}
2 \leq a_{t+1} \leq \min \left\{\frac{n \left( p^r \right)^{\lfloor \frac{\ell}{2r}\rfloor}}{p^\ell -1}, n \right\}.
\end{equation}
Let $ s \in \{0,1 , \ldots, z\}$ such that $s < t$. Then we can construct an
asymmetric EAQECC with parameters
$$
\bigg[\bigg[n, n - \frac{\ell}{r} \big\lceil \big(a_{t+1} -1\big) \big(1- \frac{1}{q}\big)\big\rceil -1,\\
 (a_{t+1}+1/a_{s+1}+1); \frac{\ell}{r} \big\lceil \big(a_{s+1} -1\big) \big(1- \frac{1}{q}\big)\big\rceil+1\bigg]\bigg]_q.
$$
\end{cor}
\begin{proof}
It follows from the proof of \cite[Theorem 10]{Akk} where it is showed that if Inequalities (\ref{c1}) and (\ref{c2}) hold, then the number $t$ of non-zero cyclotomic cosets considered is
\[
\left\lceil (a_{t+1} -1) (1- \frac{1}{q})\right\rceil
\]
and all of them have cardinality $\ell/r$.
\end{proof}

\begin{rem}
Notice that the parameters of the asymmetric EAQECC given in Corollary \ref{elc} can be written as follows:
$$
\bigg[\bigg[n, n - \frac{\ell}{r} t -1, (a_{t+1}+1/a_{s+1}+1); \frac{\ell}{r} s +1 \bigg]\bigg]_q.
$$
\end{rem}

\section{Entanglement and Minimum Distances}
\label{sect45}
Assume that $C$ is a (standard) asymmetric quantum code over a field $\mathbb{F}_q$ coming from the CSS construction with parameters $[[n,k,d_z/d_x]]_q$. Considering entanglement and a suitable extension of constituent codes, it is possible to increase the value of $d_z$ (or $d_x$) keeping the length and information rate. Therefore, one may increase the asymmetry ratio (the ratio between $d_z$ and $d_x$). Indeed, for both the standard case and the entanglement-assisted case, one considers two linear codes $C_1$ and $C_2$ and it holds  that $d_z = \mathrm{wt}(C_1^\perp \setminus (C_2 \cap C_1^\perp))$ and $d_x= \mathrm{wt}(C_2^\perp \setminus (C_1\cap C_2^\perp))$. However, for the standard case it must be imposed  that $C_2 \subseteq C_1^\perp$. Thus, given a pair of codes $C_1, C_2$ such that $C_2 \subseteq C_1^\perp$, we may consider a new pair of linear codes, $C'_1$ and $C'_2$, by enlarging $C_1$ or $C_2$, in such a way that either $C'_2 \subseteq (C'_1)^\perp$ or $C'_1 \subseteq C_2^\perp$ do not hold any more, but this new pair gives an asymmetric EAQECC with better parameters. Hence, taking into account that $c= \dim C'_1 - \dim \left( C'_1 \cap (C'_2)^\perp \right) = \dim C'_2 - \dim \left( C'_2 \cap (C'_1)^\perp \right)$, the information rate is kept, one of the minimum distances is the same and the other one increases.

Let us illustrate the above technique with a small example. Keep the notation as in Section \ref{sect4} and assume that $\ell=r$, that is we do not consider subfield-subcodes in this example. Let $C_i$ be the code $E_{\Delta_i}$, for $1 \leq i \leq 2$, where $\Delta_1 = \{2\}$ and $\Delta_2=\{0,n-1\}$. Set $C'_1= E_{\Delta'_1}$ with $\Delta'_1=\{1,2\}$. Then $C_2 \subseteq C_1^\perp=E_{\Delta_1^\perp}$, with $\Delta_1^\perp=\{0,1, \ldots,n-3,n-1\}$. Now setting $C'_2=C_2$, we deduce that the value $d_x$ for the standard case $(C_1,C_2)$ and the entanglement-assisted case $(C'_1,C'_2)$ is the same. However, in the standard case,
$$
d_z= \mathrm{wt} (C_1^\perp \setminus (C_2 \cap C_1^\perp)) =2,
$$
because the cardinality of  $\Delta_1$ is one.  But
$$ (C'_1)^\perp \setminus \left(C'_2 \cap (C'_1)^\perp\right) \subseteq \left(C'_1\right)^\perp
$$ and then, when one considers entanglement, $$d_z= \mathrm{wt} \left((C'_1)^\perp \setminus \left(C'_2 \cap (C'_1)^\perp\right)\right) \geq 3$$ by the BCH bound. As a consequence, to share entanglement allows us to increase the value $d_z$ and therefore the asymmetry ratio.

\section{Examples of Asymmetric EAQECC}
\label{sect5}

Table \ref{tabla1} shows the different values involved in the construction of asymmetric EAQECCs, over several finite fields, constructed as in Theorem \ref{elB}. The last two columns display the representatives of the cyclotomic cosets used to define the codes $C_1$ and $C_2$. Notice that the parameters of our asymmetric EAQECCs follow immediately from Theorem \ref{el4} and are $[[n, n-k_1, d_z/d_x;c]]_q$. All these codes exceed the asymmetric Gilbert-Varshamov bound proved in Theorem \ref{FGV}. In addition, for fixed values $(q,n,k_1,k_2,c)$, we consider the set $P$ of pairs $(d_1,d_2)$ of Z-minimum and X-minimum distances of asymmetric EAQECCs such that $(d_1,d_2)$  does not exceed the bound in Theorem \ref{FGV} but either $(d_1+1,d_2)$ or $(d_1,d_2+1)$ beat it; we have noticed that frequently the cardinality of $P$ is one. Let $(\mathfrak{d}_z,\mathfrak{d}_x)$ be the maximum of $P$ with respect to the lexicographical order where $(1,0) > (0,1)$. Table \ref{tabla1} displays the threshold $(\mathfrak{d}_z,\mathfrak{d}_x)$ as well. Note that many times our codes exceed both values $\mathfrak{d}_z$ and $\mathfrak{d}_x$ by one or two units.

\begin{tiny}

\begin{table*}[h!]
\caption{Asymmetric EAQECC coming from Theorem \ref{elB}}
\begin{tabular}{lllllllllp{0.15\linewidth}p{0.3\linewidth}}
$q$&$n$&$k_1$&$k_2$&$c$&$d_z=d(C_1^\perp)$&$d_x=d(C_2^\perp)$&$(\mathfrak{d}_z, \mathfrak{d}_x)$& Cyclotomic Cosets Defining $C_1$& Cyclotomic Cosets\newline Defining $C_2$\\\hline
4&15&3&1&1&3&2&(2,1)&$\{0,1\}$& $\{0\}$\\
5&24&5&3&3&4&3&(2,2)&$\{0,1,2\}$& $\{0, 23\}$\\
7&19&7&4&4&5&3&(4,2)&$\{0,1,2\}$&$\{0, 18\}$ \\
7&19&13&10&10&9&6&(8,6)&$\{0,1,2,4,5\}$& $\{0,15,17,18\}$\\
8&63&7&1&1&5&2&(3,1)&$\{0,1,2,3\}$&$\{0\}$ \\
8&63&11&3&3&7&3&(5,2)&$\{0,1,2,3,4,5\}$& $\{0,62\}$\\
9&40&10&5&5&7&4&(5,3)&$\{0,1,2,3,4,5\}$& $\{0,38,39\}$ \\
9&40&12&3&3&8&3&(6,2)&$\{0,1,2,3,4,5,6\}$&$\{0,39\}$ \\
9&40&12&7&7&8&5&(6,3)&$\{0,1,2,3,4,5,6\}$& $\{0,37,38,39\}$\\
16&51&9&3&3&6&3&(5,2)&$\{0,1,2,3,4\}$& $\{0,50\}$ \\
16&51&11&1&1&7&2&(6,1)&$\{0,1,2,3,4\}$& $\{0,50\}$ \\
16&51&11&3&3&7&3&(6,2)&$\{0,1,2,3,4,5\}$& $\{0,50\}$\\
16&51&17&5&5&10&4&(10,3)&$\{0,1,2,3,4,5,6,7,8\}$& $\{0,49,50\}$ \\
16&51&19&5&5&12&4&(11,3)&$\{0,1,2,3,4,5,6,7,8,9\}$& $\{0,49,50\}$\\
16&51&23&3&3&15&3&(14,2)&$\{0,1,2,3,4,5,6,7,8,9,11,12\}$& $\{0,50\}$ \\
16&51&23&9&9&15&6&(14,5)&$\{0,1,2,3,4,5,6,7,8,9,11,12\}$& $\{0,47,48,49,50\}$\\
16&51&27&5&5&18&4&(17,3)&$\{0,1,2,3,4,5,6,7,8,9,11,12,14,15\}$& $\{0,49,50\}$ \\
25&48&6&4&4&5&4&(4,2)&$\{0,1,2,3\}$&$\{0,46,47\}$\\
25&48&10&4&4&8&4&(6,2)&$\{0,1,2,3,4,5,6\}$&$\{0,46,47\}$ \\
25&48&10&7&7&8&6&(6,4)&$\{0,1,2,3,4,5,6\}$&$\{0,44,45,46,47\}$\\
25&48&12&3&3&9&3&(7,2)&$\{0,1,2,3,4,5,6,7\}$&$\{0,47\}$\\
25&48&12&6&6&9&5&(7,4)&$\{0,1,2,3,4,5,6,7\}$&$\{0,45,46,47\}$\\
\end{tabular}
\label{tabla1}
\end{table*}
\end{tiny}

We would like to add that our construction from BCH codes (regarded as in Section \ref{sect4}), using Theorem \ref{el4}, may produce good symmetric EAQECCs as well. We use the fact that an asymmetric EAQECC provides a symmetric EAQECC with the same parameters but its minimum distance, which is the minimum of the two minimum distances  $d_x$ and $d_z$. In all our examples, both minimum distances are equal. Thus, Table \ref{tabla2}  displays values of codes coming from our construction, giving rise to symmetric EAQECCs  whose parameters are $[[n, k=n-k_1-k_2+c, d=d_z=d_x;c]]_q$. We have used the Hartmann-Tzeng bound \cite{pel} in our computations. All codes in this table exceed the Gilbert-Varshamov bound for symmetric EAQECCs \cite{GaHMR}. Table \ref{tabla2} also contains, for each code $C$, the minimum distance $\mathfrak{d}$ of a symmetric EAQECC  with the same parameters $(q,n,k,c)$ as $C$ and such that $\mathfrak{d}$ does not beat the above mentioned symmetric Gilbert-Varshamov bound but $\mathfrak{d} +1$ does. In several cases, our codes exceed the value $\mathfrak{d}$ by more than one unit.

We conclude this section by showing another sign of the goodness of our codes and from the advantages of considering entanglement. Table \ref{tabla1} provides two asymmetric EAQECCs with parameters $[[24,19,4/3;c=3]]_5$, and $[[15,12,3/2;c=1]]_4$, which (using entanglement) have better parameters that the optimal (non entanglement-assisted) asymmetric QECCs $[[24,17,4/3]]_5$ and $[[15,11,3/2]]_4$ given in \cite{EZ2}. Finally, Table \ref{tabla2} shows two binary symmetric EAQECCs with parameters $[[15,4,8; c=8]]_2$ and $[[31,25,4; c=6]]_2$ with better parameters than the best  (non entanglement-assisted) QECCs $[[15,4,4]]_2$ and $[[31,25,2]]_2$ given in \cite{grassl}.

\begin{tiny}
\begin{table*}[h!]
\caption{Symmetric EAQECC coming from Theorem \ref{elB}}
\begin{tabular}{lllllllllp{0.25\linewidth}p{0.3\linewidth}}
$q$&$n$&$k_1$&$k_2$&$c$&$d_z=d(C_1^\perp)$&$d_x=d(C_2^\perp)$& $\mathfrak{d}$& Cyclotomic Cosets Defining $C_1$& Cyclotomic Cosets Defining $C_2$\\ \hline
2 & 15 & 11 & 11 & 11 & 8 & 8 & 6& $\{0,1,3,5\}$& $\{0,1,3,5,7 \}$\\
2 & 15 & 10 & 10 & 10 & 7 & 7 & 6 &$\{1,3,5\}$&$\{3,5,7\}$\\
2 & 15 & 9 & 9 & 9 & 6 & 6 & 5 & $\{0,1,3 \}$ &$\{0,3,7 \}$\\
2 & 15 & 7 & 5 & 1 & 4 & 4 & 3 & $\{0,1,5 \}$& $\{0,1 \}$\\
2 & 15 & 10 & 10 & 6 & 7 & 7 & 6 & $\{1,3,5 \}$ &$\{1,3,5 \}$\\
2 & 15 & 14 & 14 & 14 & 15 & 15 & 11 &$\{1,3,5,7 \}$&$\{1,3,5,7 \}$\\
2 & 15 & 13 & 13 & 13 & 10 & 10 & 9 & $\{0,1,3,7\}$ &$\{0,1,3,7 \}$\\
2 & 31 & 6 & 6 & 6 & 4 & 4 & 2 & $\{0,1 \}$ & $\{0,15 \}$\\
2 & 31 & 11 & 11 & 1 & 6 & 6 & 5 &$\{0,1,3 \}$ & $\{0,1,3 \}$\\
2 & 31 & 5 & 5 & 5 & 3 & 3 & 2 &$\{ 1\}$ & $\{15 \}$\\
2 & 31 & 21 & 21 & 16 & 12 & 12 & 11 & $\{0,1,5,7,15 \}$ & $\{0,3,7,11,15 \}$\\
2 & 31 & 20 & 20 & 15 & 11 & 11 & 10 & $\{1,5,7,15 \}$ & $\{3,7,11,15 \}$\\
2 & 31 & 10 & 10 & 10 & 5 & 5 & 4 & $\{1,3\}$ & $\{7,15 \}$\\
2 & 63 & 7 & 7 & 7 & 4 & 4 & 2 & $\{0,1 \}$ & $\{0,31 \}$\\
2 & 63 & 13 & 13 & 13 & 6 & 6 & 5 & $\{0,1,3 \}$ & $\{0, 15, 31 \}$\\
2 & 63 & 9 & 7 & 7 & 4 & 4 & 3 & $\{0,1,21 \}$ & $\{0, 31 \}$\\
2 & 63 & 10 & 7 & 7 & 4 & 4 & 3 & $\{0,1,9 \}$ & $\{0, 31 \}$\\
2 & 63 & 15 & 13 & 13 & 6 & 6 & 5 & $\{0,1,3,21 \}$ & $\{0, 15, 31 \}$\\
2 & 63 & 19 & 19 & 19 & 8 & 8 & 7 & $\{0,1,3,5 \}$ & $\{0, 15,23,31 \}$\\
3 & 26 & 7 & 7 & 1 & 5 & 5 & 4 & $\{0,1,2 \}$ & $\{0, 7, 14 \}$\\
3 & 26 & 18 & 18 & 18 & 13 & 13 & 12 & $\{1,2,4,5,7,8 \}$ & $\{2,5,7,8,14, 17 \}$\\
4 & 15 & 4 & 4 & 1 & 4 & 4 & 3 & $\{0,1,5 \}$ & $\{0,1,5 \}$\\
4 & 15 & 8 & 8 & 3 & 7 & 7 & 6 & $\{0,1,2,3,5 \}$ & $\{0,1,2,3,5 \}$\\
4 & 15 & 11 & 11 & 9 & 10 & 10 & 9 & $\{1,2,3,5,6,7 \}$ & $\{1,2,3,6,10,11 \}$\\
4 & 17 & 4 & 4 & 4 & 4 & 4 & 3 & $\{ 6\}$ & $\{6 \}$\\
4 & 17 & 8 & 8 & 4 & 7 & 7 & 5 & $\{1,3\}$ & $\{1,6\}$\\
4 & 17 & 13 & 13 & 13 & 12 & 12 & 10 & $\{0,1,2,3 \}$ & $\{0,1,2,3 \}$\\
4 & 17 & 16 & 16 & 16 & 17 & 17 & 14 & $\{1,2,3,6 \}$ & $\{1,2,3,6 \}$\\
4 & 17 & 9 & 8 & 8 & 7 & 7 & 6 &$\{0,1,3 \}$ & $\{1,3 \}$\\
5 & 24 & 4 & 4 & 4 & 4 & 4 & 3 &$\{0,1,6 \}$ & $\{0, 18, 19 \}$\\
5 & 24 & 4 & 4 & 1 & 4 & 4 & 3 & $\{0,1,6 \}$ & $\{12, 13, 18 \}$\\
5 & 24 & 10 & 10 & 4 & 8 & 8 & 7 & $\{0,1,2,3,4,6 \}$ & $\{2,6,8,9, 12, 19 \}$\\
5 & 24 & 5 & 4 & 4 & 4 & 4 & 3 & $\{0,1,6,12 \}$ & $\{0, 18,19 \}$\\
5 & 24 & 20 & 20 & 20 & 18 & 18 & 16 &$\{0,1,2,3,4,7,8,9,13,14,18 \}$ & $\{0,2,3,4,6,7,8,9,13,14,19 \}$\\
\end{tabular}
\label{tabla2}
\end{table*}
\end{tiny}

\section{Conclusion}
In this article we show how to construct asymmetric EAQECCs. That is, entanglement-assisted quantum error-correcting codes designed for the case when phase-shift errors happen more likely than qudit-flip errors, as it is with the combined amplitude damping and dephasing channel. Moreover, they can be constructed from any pair of classical linear codes since the encoder and decoder may share entanglement. Following our framework, a concrete construction using BCH codes is proposed and we expect further
families of asymmetric EAQECCs to be proposed. In particular, we will extend the BCH construction by considering evaluation of polynomials in several variables which will hopefully give better
results and a larger constellation of codes. The
Gilbert-Varshamov-type bound provided in this article will allow
researchers to check the goodness of the parameters of codes of this type.

\section{Acknowledgement}
We would like to thank the reviewers for their insightful comments.


\begin{thebibliography}{00}

\bibitem{Akk} S.A. Aly, S. Klappenecker and P.K. Sarvepalli, ``On quantum and classical BCH codes," {\it IEEE Trans. Inf. Theory} {\bf 53} (2007) 1183-1188.


\bibitem{7kkk} A. Ashikhmin, A. Barg, E. Knill and S. Litsyn, ``Quantum error-detection I: Statement of the problem,'' {\it IEEE Trans. Inf. Theory} {\bf 46} (2000) 778-788.

\bibitem{8kkk} A. Ashikhmin, A. Barg, E. Knill and S. Litsyn, ``Quantum error-detection II: Bounds,'' {\it IEEE Trans. Inf. Theory} {\bf 46} (2000) 789-800.

\bibitem{3DC} A. Ashikhmin and E. Knill, ``Upper bounds on the size of quantum codes,'' {\it IEEE Trans. Inf. Theory} {\bf 45} (1999) 1206-1215.

\bibitem{AK} A. Ashikhmin and E. Knill, ``Non-binary quantum stabilizer codes,'' {\it IEEE Trans. Inf. Theory} {\bf 47} (2001) 3065-3072.

\bibitem{BE} J. Bierbrauer and Y. Edel, ``Quantum twisted codes,'' {\it J. Comb. Designs} {\bf 8} (2000) 174-188.

\bibitem{brun1}
T. Brun, I. Dvetak and M.~H. Hsieh.
\newblock ``Correcting quantum codes with entanglement,''
\newblock {\em Science} {\bf 314} (2006) 436-439.

\bibitem{18kkk} A.R. Calderbank, E.M. Rains, P.W. Shor and N.J.A Sloane, ``Quantum error correction and orthogonal geometry,'' {\it Phys. Rev. Lett.} {\bf 76} (1997) 405-409.


\bibitem{19kkk}  A.R. Calderbank, E.M. Rains, P.W. Shor and N.J.A Sloane, ``Quantum error correction via codes over GF(4),'' {\it IEEE Trans. Inf. Theory} {\bf 44} (1998) 1369-1387.

\bibitem{20kkk} A.R. Calderbank and P.W. Shor, ``Good quantum error-correcting codes
exist,'' {\it Phys. Rev. A} {\bf 54} (1996) 1098-1105.

\bibitem{cas} I. Cascudo, ``On squares of cyclic codes,'' {\it IEEE Trans. Inf. Theory} {\bf 65} (2019) 1034-1047.

\bibitem{8AS} D. Dieks, ``Communication by EPR devices,'' {\it Phys. Rev. A} {\bf 92} (1982) 271.

\bibitem{eck} A. Ekert and C. Macchiavello, ``Quantum error correction for communication,'' {\it Phys. Rev. Lett.} {\bf 77} (1996) 2585.

\bibitem{EZ1} M.F. Ezerman, S. Jitman, H.M. Kiah and S. Ling, ``Pure asymmetric quantum MDS codes from CSS construction: a complete characterization," {\it Int. J. Quantum Inf.} {\bf 11} (2013) 1350027.

\bibitem{EZ2} M.F. Ezerman, S. Jitman, H.M. Kiah and S. Ling, ``CSS-like constructions of asymmetric quantum codes,"  {\it IEEE Trans. Inf. Theory} {\bf  59} (2013) 6732-6754.

\bibitem{EZ3} M.F. Ezerman, S. Ling and P. Sol\'e, ``Additive asymmetric quantum codes,'' {\it IEEE Trans. Inf. Theory} {\bf 57} (2011) 5536-5550.

\bibitem{feng} K. Feng and Z. Ma, ``A finite Gilbert-Varshamov bound for pure stabilizer quantum codes,'' {\it IEEE Trans. Inf. Theory} {\bf 50} (2004) 3323-3325.

\bibitem{GaGHR} C. Galindo, O. Geil, F. Hernando and D. Ruano, ``On the distance of stabilizer quantum codes from J-affine variety codes,'' {\it Quantum Inf.   Process.} {\bf 16} (2017) 111.

\bibitem{GeG} C. Galindo, O. Geil, F. Hernando and D. Ruano, ``Improved constructions of nested codes,'' {\it IEEE Trans. Inf. Theory} {\bf 64} (2018) 2444-2459.

\bibitem{GaH} C. Galindo and F. Hernando, ``Quantum codes from affine variety codes and their subfield-subcodes,'' {\it Des. Codes Cryptogr. } {\bf 76} (2015) 76-89.

\bibitem{GaHMR} C. Galindo, F. Hernando, R. Matsumoto and D. Ruano, ``Entanglement-assisted quantum error-correcting codes over arbitrary finite fields,'' {\it Quantum Inf.   Process.} {\bf 18} (2019) 116.

\bibitem{38kkk} D. Gottesman, ``A class of quantum error-correcting codes saturating
the quantum Hamming bound,'' {\it Phys. Rev. A} {\bf 54} (1996) 1862-1868.

\bibitem{grassl} M. Grassl,
``Bounds on the minimum distance of linear codes and quantum codes,"
Online available at http://www.codetables.de. Accessed on 2019-12-22.

\bibitem{opt} M. Grassl, T. Beth, T. and M. R\"{o}tteler, ``On optimal quantum codes,'' {\it Int. J. Quantum Inform.} {\bf 2} (2004) 757-775.

\bibitem{45kkk}  M. Grassl and M. R\"{o}tteler, ``Quantum BCH codes,'' In Proc. X Int. Symp. Theor.  Elec. Eng. Germany (1999) 207-212.

\bibitem{hsieh}
M.-H. Hsieh, I. Dvetak and T. Brun,
\newblock ``General entanglement-assisted quantum error-correcting codes,''
\newblock {\em Phys.\ Rev.\ A} 76 (2007) 062313.

\bibitem{IOF} L. Ioffe and M. M\'ezart, ``Asymmetric  quantum error-correcting codes,''   {\it Phys. Rev. A, Gen. Phys.} {\bf 75} (2007) 032345.

\bibitem{ketkar06}
A.~Ketkar, A.~Klappenecker, S.~Kumar and P.~K. Sarvepalli,
\newblock ``Nonbinary stabilizer codes over finite fields,''
\newblock {\em IEEE Trans.\ Inform.\ Theory} {\bf 52} (2006) 4892-4924.

\bibitem{lag3} G.G. La Guardia, ``Construction of new families of nonbinary quantum BCH codes,'' {\it Phys. Rev. A} {\bf 80} (2009) 042331.


\bibitem{LG1} G.G. La Guardia, ``Asymmetric quantum product codes,'' {\it Int. J. Quantum Inf.} {\bf  10} (2012) 1250005.

\bibitem{LG2} G.G. La Guardia, ``Asymmetric quantum Reed-Solomon and generalized Reed-Solomon codes,'' {\it Quantum Inf. Process.} {\bf 11} (2012) 591-604.

\bibitem{LG3} G.G. La Guardia,  ``Asymmetric quantum codes: new codes from old,'' {\it Quantum Inf. Process.} {\bf 12} (2013) 2771-2790.

\bibitem{LG4} G.G. La Guardia, ``On the construction of asymmetric quantum codes,'' {\it Internat. J. Theoret. Phys.} {\bf 53} (2014) 2312-2322.

\bibitem{schin} L. Luo, Z. Ma, Z. Wei and R. Leng,
\newblock ``Non-binary entanglement-assisted quantum stabilizer codes,''
\newblock {\em Sci. China Inf. Sci.} {\bf 60} (2017) 42501.

\bibitem{Matsu} R. Matsumoto, ``Two Gilbert-Varshamov type existential bounds for asymmetric quantum error correcting bounds'' {\it Quantum Inf. Process.} {\bf 16} (2017) 285.

\bibitem{71kkk} R. Matsumoto and T. Uyematsu, ``Constructing quantum error correcting codes for $p^m$ state systems from classical error correcting codes'' {\it IEICE Trans. Fund.} {\bf E83-A} (2000) 1878-1883.

\bibitem{matsumotouematsu01}
R.~Matsumoto and T.~Uyematsu,
\newblock ``Lower bound for the quantum capacity of a discrete memoryless quantum
  channel,''
\newblock {\em J. Math. Phys.} {\bf 43} (2002) 4391-4403.

\bibitem{pel}
R. Pellikaan,
\newblock ``The shift bound for cyclic, Reed-Muller and
geometric Goppa codes,"
\newblock in Arithmetic, Geometry and Coding Theory, Walter de Gruyter, pp. 155-174, 1996.

\bibitem{28DC} E.M. Rains ``Nonbinary quantum codes,'' {\em IEEE Trans.\ Inform.\ Theory} {\bf 45} (1999) 1827-1832.

\bibitem{Sarve} P.K. Sarvepalli, A.  Klappenecker and  M. R\"{o}tteler, ``Asymmetric quantum codes: constructions, bounds and performance,'' {\it Proc. R. Soc. Lond. Ser. A Math. Phys. Eng. Sci.} {\bf 465} (2009) 1645-1672.

\bibitem{22RBC} P.W. Shor, ``Polynomial-time algorithms for prime factorization and discrete logarithms on a quantum computer,'' in Proc. 35th  ann. symp. found. comp. sc., {\it IEEE Comp. Soc. Press} (1994) 124-134.

\bibitem{23RBC} P.W. Shor, ``Scheme for reducing decoherence in quantum computer memory,'' {\it Phys. Rev. A} {\bf 52} (1995) 2493-2496.

\bibitem{95kkk} A.M. Steane, ``Simple quantum error correcting codes,'' {\it Phys. Rev. Lett.} {\bf 77} (1996) 793-797.

\bibitem{wilde08}
M.~M. Wilde and T.~A. Brun,
\newblock ``Optimal entanglement formulas for entanglement-assisted quantum
  coding,''
\newblock {\em Phys.\ Rev.\ A}  {\bf 77} (2008) 064302.

\bibitem{26RBC} W.K. Wootters and W.H. Zurek, ``A single quantum cannot be cloned,'' {\it Nature}
{\bf 299} (1982) 802-803.

\end{thebibliography}
\end{document}